\documentclass[preprint]{elsarticle}

\usepackage{array}
\usepackage{float}

\usepackage{amsmath} 
\usepackage{amssymb}  
\usepackage{theorem}

\def\S{S_{\md, \md'}}
\def\St{\tilde{S}_{\md, \md'}}
\def\md{\mathbf{d}}
\def\E{\mathcal{E}}
\def\e{\epsilon}
\def\d{\delta}
\def\dt{\tilde{d}}
\def\ab{($\e$,\,$\d$)}
\def\Prob{\mathbb{P}}
\newcommand{\A}{\mathcal{A}}

\newtheorem{theorem}{Theorem}
\newtheorem{corollary}{Corollary}
\newtheorem{definition}{Definition}
\newtheorem{proposition}{Proposition}
\newtheorem{example}{Example}
\DeclareMathOperator{\h}{h}

\begin{document}

\title{Differentially Private Response Mechanisms on Categorical Data}
\author[naoise_doug]{Naoise Holohan}
\author[naoise_doug]{Douglas J. Leith}
\author[ollie]{Oliver Mason\corref{cor1}}

\address[naoise_doug]{School of Computer Science and Statistics, Trinity College Dublin, Ireland}
\address[ollie]{Dept. of Mathematics and Statistics/Hamilton Institute, Maynooth University-National University of Ireland
Maynooth, \\Co. Kildare, Ireland}
\cortext[cor1]{Corresponding author. Tel.: +353 (0)1 7083672; fax: +353
5(0)1 7083913; email: oliver.mason@nuim.ie}

\begin{abstract}
We study mechanisms for differential privacy on finite datasets.  By deriving \emph{sufficient sets} for differential privacy we obtain necessary and sufficient conditions for differential privacy, a tight lower bound on the maximal expected error of a discrete mechanism and a characterisation of the optimal mechanism which minimises the maximal expected error within the class of mechanisms considered.  
\end{abstract}


\maketitle
\section{Introduction}

Data privacy has been of interest to researchers for decades \cite{DJL79}, but high-profile privacy breaches in recent years, such as those involving AOL \cite{BZ06} and Netflix \cite{NS08}, have renewed focus on the topic.  The movement towards smart metering systems for electricity, water and other utilities and the greater use of data mining in so-called smart cities and transport have given rise to further concerns over personal data privacy.  

The most traditional framework for the study of data privacy is that of tabular data.  A simple model of this type considers the data to be arranged as individual records within a table, where each record contains entries from some underlying dataset, which may be continuous or discrete depending on the type of data being studied.  Simple anonymisation techniques, such as removing names and social security numbers (so-called unique identifiers) from the data, have been shown to be inadequate \cite{Swe02}.   More sophisticated frameworks such as $k$-anonymity \cite{Swe02} and $\ell$-diversity \cite{MKG07} are also vulnerable to privacy attacks via the use of appropriate side-information or data from external sources \cite{MKG07, LLV07}.

Within the last decade, differential privacy \cite{Dwo06} has emerged as a popular framework for research in the field of data privacy based on its capability to provide a quantifiable basis for privacy preserving data publishing and mining.  This is a probabilistic approach to data privacy in which a suitably \emph{randomised} version of the correct response to a query is released.  The core idea is founded on the simple premise that the response to a user query should not be too tightly coupled with any one entry in the table.  One widely-adopted implementation of differential privacy for real-valued databases is to add an appropriate amount of noise sampled from a Laplace distribution to each cell of the database \cite{Dwo08}.

Much research on differential privacy to date has been completed on real-valued databases \cite{Dwo08}, although a considerable body of literature also exists on discrete data \cite{MT07, CMF11}; in particular some recent work has focussed on graph data relevant to applications in areas such as social networks \cite{KRS11, BBD13}.  

Differentially private mechanisms can be divided into two distinct classes: sanitisation based mechanisms; and output perturbation based mechanisms.  Our concern here is with the former class, which first constructs a sanitised version of the database and then answers queries on this.  It has been shown in \cite{HLM15} that if the sanitised database satisfies the requirements of differential privacy, then any query can be answered on it in a differentially private manner.

In writing this paper, we have two aims: the first is to present a set of new results on the mathematical foundations of differential privacy for discrete data; the second is to bring the problems in this field to the attention of researchers in discrete applied mathematics.  

We examine differentially private mechanisms for discrete data within the general probabilistic framework described in our previous paper \cite{HLM15}.  As we deal with finite datasets here, many of the measure-theoretic details required for the more general setting can be suppressed.  However, to properly set context, we include the more general definitions here in Section \ref{sc:prelim}.  

Our first results concern an adaptation for discrete data of the exponential mechanism introduced by McSherry and Talwar.  In particular, we consider the problem of \emph{sufficient sets} for differential privacy for this mechanism.  This problem is motivated by the practical issue of testing whether or not a mechanism is differentially private and arises from the following simple considerations. 

For a sanitisation to be differentially private, certain inequalities (described formally later) must hold on all subsets of the database space, which can necessitate checking a prohibitively large collection of sets in order to test for differential privacy.  The question of sufficient sets asks whether it is sufficient for the differential privacy condition to hold on a collection of these subsets for it to hold on all subsets. We can therefore reduce the workload required to check that a mechanism satisfies differential privacy.  In Section \ref{sc:dem}, we present results characterising sufficient sets for the discrete exponential mechanism.  We then use these to give necessary and sufficient conditions for differential privacy for this mechanism.  

A major concern of privacy research is the trade-off between privacy and accuracy.  For the current setting, in the absence of a given metric on the dataset, we measure the error of a sanitisation using hamming distance; in Theorem \ref{th:demerr} we derive a tight lower bound on the maximal expected error of a discrete exponential mechanism.  

In Section \ref{sc:san} we consider a seemingly unrelated approach to database sanitisation: product sanitisations.  We show that these are in fact equivalent to the discrete exponential mechanism constructed using the hamming distance and, building on results in \cite{HLM15}, we characterise differential privacy and the error for these in Theorems \ref{th:uni} and \ref{th8} respectively.  Finally in Theorem \ref{th:opt} we provide a characterisation of the optimal product sanitisation mechanism, which minimises the maximal expected error within the class of product sanitisations (and hence within the class of discrete exponential mechanisms).  Concluding remarks are given in Section~\ref{sc:conc}.

\subsection{Related work}

Before the advent of differential privacy, Fienberg examined the use of data swapping and cell supression for privacy protection on categorical data \cite{FMS98}. Dwork then presented the notion of differential privacy in \cite{Dwo06}, and it limitations were discussed by Dankar in \cite{DE12}, including its applicability to categorical data.

Dwork's work was closely followed by McSherry and Talwar who proposed the exponential mechanism in \cite{MT07}. An instantiation of this was used by Hardt and Talwar \cite{HT10} in examining the geometry of differential privacy. Mohammed made use of the exponential mechanism for releasing count queries in \cite{MCF11,CMF11}, while a more recent contribution has looked at differential privacy on counts using a combination of the Laplace and exponential mechanisms \cite{ZM14}.

\section{Preliminaries}\label{sc:prelim}

\subsection{Database model}
We consider a finite data set $D$ with $(m+1)$ elements ($m\ge1$). A database $\md$ with $n$ rows drawn from this data set is represented by a vector $\md=(d_1, \dots, d_n)\in D^n$. $D$ is equipped with a $\sigma$-algebra, in this case the power set $2^D$ and $D^n$ inherits the product $\sigma$-algebra, $2^{D^n}$. We are therefore considering all subsets of $D$ and $D^n$.

We will consider hamming distance on $D^n$. Recall that the hamming distance, $\h: D^n\times D^n\to \{0,1,\dots,n\}$, between two databases is the number of rows on which they differ:
\begin{equation}\label{eq:ham}
\h(\md, \md')=|\{i:d_i\ne d_i^\prime\}|.
\end{equation}

We say two databases $\md, \md'\in D^n$ are neighbours, written $\md\sim\md'$, if $\h(\md, \md')=1$, \emph{i.e.} they differ on exactly one row.

\subsection{Query model}\label{sc:query}
We make use of the generalised query model introduced in \cite{HLM15}, adapted to the discrete setting. A query $Q:D^n\to E_Q$ outputs a response in $E_Q$, the structure of which is not specified (it may be numeric, categorical, functional, etc). $E_Q$ is, however, equipped with a $\sigma$-algebra $\A_Q$. We require that all queries be measurable, which is trivial in this setting since $Q^{-1}(A)\subseteq D^n$ for all $A\in\A_Q$.

\subsection{Response mechanism}
Let $(\Omega, \mathcal{F}, \Prob)$ be a probability space. We define a response mechanism for a query $Q$ to be the family of measurable mappings
\begin{equation}
\{X_{Q,\md}:\Omega\to E_Q\mid\md\in D^n\}.
\end{equation}
For simplicity, when the query in question is the identity query $I$, we denote $X_{I, \md}$ by $X_\md:\Omega\to D^n$. Where there is no ambiguity, the response mechanism will be written as $\{X_{Q, \md}\}$ (or $\{X_\md\}$ when dealing with the identity query).

In this paper we deal exclusively with \emph{sanitised response mechanisms}, where $X_{Q, \md} = Q\circ X_\md$. In this case the mechanism is generated by first sanitising the database $\md$ and then answering queries on the sanitised database $X_\md$ without any further modification to the data. Hence, a sanitised response mechanism for a query $Q$ is the family of mappings
\begin{equation}
\{X_{Q,\md}=Q\circ X_\md:\Omega\to E_Q\mid\md\in D^n\}.
\end{equation}

One mechanism which we will make use of in this paper is the exponential mechanism, as described by McSherry and Talwar \cite{MT07}. The following is the exponential mechanism written in our notation, reformulated to deal specifically with discrete data.

\begin{definition}[Exponential mechanism]\label{df:mt07}
Given a query $Q$, a query output space $E_Q$, a utility function (which measures the utility of all possible query answers to the database being queried) $u:D^n\times E_Q\to\mathbb{R}$, a measure $\mu:E_Q\to\mathbb{R}$ and a normalisation constant $C_\md$, the exponential mechanism is defined to be the family of mappings $\{X_{Q,\md}:\Omega\to E_Q\mid\md\in D^n\}$, with probability density function with respect to $\mu$ given by $C_\md^{-1} e^{\epsilon u(\md,q)}$ for each $\md\in D^n$ and for $q\in E_Q$.
\end{definition}

\textbf{Comment:} If $E_Q$ is discrete, we can specify $\mu$ by $\{\mu(q):q\in E_Q\}$ and the exponential mechanism has probability mass function
\begin{equation}\label{eq:expmech}
\Prob(X_{Q,\md}=q)=C_\md^{-1} e^{\epsilon u(\md,q)}\mu(q),
\end{equation}
where $q\in E_Q$.

\subsection{Differential privacy}\label{sc:dpintro}

We now define what it means for a response mechanism in our framework to be differentially private.

\begin{definition}[Differential privacy with respect to a query]
Let a query $Q:D^n\to E_Q$ and parameters $\epsilon\ge0$ and $0\le\delta\le1$ be given. A response mechanism $\{X_{Q,\md}:\Omega\to E_Q\mid\md\in D^n\}$ is {\ab}-differentially private if
\begin{equation}\label{eq:dp}
\Prob(X_{Q,\md}\in A)\le e^\epsilon\Prob(X_{Q,\md'}\in A)+\delta,
\end{equation}
for all $\md\sim\md'\in D^n$ and for all measurable $A\subseteq E_Q$.

Note that the relation $\md\sim\md'$ is symmetric, so (\ref{eq:dp}) must hold when $\md$ and $\md'$ are swapped.
\end{definition}

\textbf{Note:} McSherry and Talwar showed that the exponential mechanism (Definition~\ref{df:mt07}) satisfies $2\epsilon\Delta u$-differential privacy ($\delta=0$), where $$\Delta u=\max_{\md\sim\md'\in D^n,q\in E_Q}|u(\md,q)-u(\md',q)|.$$

Theorem~4 of \cite{HLM15} simplifies the problem of checking differential privacy for sanitised response mechanisms. We now recall that theorem. 

\begin{theorem}[Identity query]
A sanitised response mechanism which is {\ab}-differentially private with respect to the identity query is {\ab}-differentially private with respect to any query $Q$.
\end{theorem}

We therefore need only examine the response mechanism for the identity query,
\begin{equation}
\{X_\md:\Omega\to D^n\mid\md\in D^n\}.
\end{equation}

\begin{example}[Categorical data I]\label{eg:prelim}
Suppose the data we are interested in records individuals' favourite hobby. The data set $D$ would contain a list of all hobbies. For simplicity in this example, we restrict answers to the following five hobbies: Sports; Cars; Television; Computer games; and Reading. Hence, $m=4$. Each database $\md$ would contain the favourite hobby of $n$ individuals. If $n=6$, one possible $\md$ could be represented by the following list: Sports; Computer games; Television; Sports; Reading; Television.

Queries on such databases could include counting the number of unique hobbies (4 in the case above) or how many list `Television' as their favourite hobby (2 in the case above). The identity query would be another valid query.
\end{example}

\section{Sufficient sets for discrete exponential mechanism}\label{sc:dem}

In order for a response mechanism to be deemed {\ab}-differentially private, (\ref{eq:dp}) must hold for all pairs of neighbouring databases and for all possible subsets of $D^n$. If we were to check all combinations, this would require checking all $nm(m+1)^n$ pairs of neighbouring databases on $2^{(m+1)^n}-2$ subsets of $D^n$ (all subsets except $D^n$ itself and $\emptyset$). Therefore, checking a discrete response mechanism for differential privacy requires $nm(m+1)^n(2^{(m+1)^n}-2)$ checks in total.

However, for a given pair of neighbouring databases, it is not always necessary to check (\ref{eq:dp}) on all subsets of $D^n$. So we ask the question: For a given $\md\sim\md'$, what is the smallest collection of subsets of $D^n$ that we need to check for (\ref{eq:dp}) to hold on all subsets of $D^n$? We call such a collection of subsets \emph{sufficient sets} of the mechanism for $\md\sim\md'$.

In this section, we examine the sufficient sets for a class of discrete response mechanisms and show that, even in the most general cases, significant improvements on workload can be made when checking for differential privacy. We also present conditions that are necessary and sufficient for differential privacy to hold. This compares to the mostly sufficient conditions presented in other differential privacy literature, which can therefore give a conservative estimate on the privacy level achieved.

\subsection{General response mechanism}
We begin by considering the exponential mechanism described by McSherry and Talwar \cite{MT07}, as detailed in Definition~\ref{df:mt07}. We wish to assign a probability to each database based on its utility to the input/reference database. This is determined by the utility function, which can be a metric or any other function deemed suitable for a particular application. As discussed in Section~\ref{sc:dpintro}, we are only concerning ourselves with the identity query.

Ordinarily, to minimise error, we would want to assign the reference/input database itself the highest probability of being returned, with decreasing likelihood the further we move away from the reference, as determined by the utility function.

\begin{definition}[Discrete exponential mechanism]\label{df:dexp}
Let $u:D^n\times D^n\to\mathbb{R}$ be given. The \emph{discrete exponential response mechanism} is defined to be a family of measurable mappings
$$\{X_\md:\Omega\to D^n\mid\md\in D^n\},$$
where each $X_\md$ satisfies
\begin{equation}
\Prob(X_\md=\md')=C_\md^{-1}e^{u(\md, \md')},
\end{equation}
for all $\md, \md'\in D^n$. As $D^n$ is finite, we can define the normalisation constant $C_\md$ as $$C_\md=\sum_{\md'\in D^n}e^{u(\md,\md')}$$ for each $\md\in D^n$.
\end{definition}

\textbf{Remark:} While similar to the exponential mechanism (\ref{eq:expmech}), the discrete exponential mechanism differs by dealing only with the identity query ($Q=I$), by having a uniform measure ($\mu=1$) and by absorbing $\epsilon$ into $u$, in order to allow consideration of {\ab}-differential privacy.

Even with this general set-up, we can still make improvements on workload when checking for differential privacy. We begin by defining the following set for each pair of neighbouring databases $\md\sim\md'\in D^n$:
\begin{align}
\S&=\left\{\md^*\in D^n \mid \Prob(X_\md=\md^*)>\Prob(X_{\md'}=\md^*)\right\}\\
&=\left\{\md^*\in D^n \mid C_\md^{-1}e^{u(\md, \md^*)}>C_{\md'}^{-1}e^{u(\md', \md^*)}\right\}\nonumber
\end{align}
This set is a collection of the ``worst-case'' databases for $\md$ and $\md'$, and, as we show in Theorem~\ref{th:dexpsuff1}, is the only set of interest when checking for differential privacy.

\begin{theorem}[Sufficient sets]\label{th:dexpsuff1}
Let $\{X_\md\}$ be a discrete exponential mechanism and fix $\md\sim\md'\in D^n$. If (\ref{eq:dp}) holds on all $A\subseteq\S$ then it will hold on all $A\subseteq D^n$.
\end{theorem}

	\begin{proof}
	We fix $\md\sim\md'\in D^n$, let $A\subseteq D^n$ be given and assume (\ref{eq:dp}) holds on all subsets of $\S$. By assmption, (\ref{eq:dp}) holds on $A_0 = A\cap\S$, hence $\Prob(X_\md\in A_0)\le e^\epsilon \Prob(X_{\md'}\in A_0)+\delta$. If $A_0=A$ we are done, so assume $A\setminus A_0\ne\emptyset$.

	For each $\md^*\in A\setminus A_0$, $\Prob(X_\md=\md^*)\le\Prob(X_{\md'}=\md^*)$. Pick one such $\md^*_0\in A\setminus A_0$, then
	\begin{align*}\Prob(X_\md\in A_0\cup \{\md^*_0\}) &= \Prob(X_\md \in A_0) + \Prob(X_\md=\md^*_0)\\
	&\le e^\epsilon \Prob(X_{\md'}\in A_0) + \delta + \Prob(X_\md=\md^*_0)\\
	&\le e^\epsilon \Prob(X_{\md'}\in A_0) + \delta + \Prob(X_{\md'}=\md^*_0)\\
	&\le e^\epsilon \Prob(X_{\md'}\in A_0\cup \{\md^*_0\}) + \delta.
	\end{align*}

	Hence, (\ref{eq:dp}) holds on $A_1=A_0\cup \{\md^*_0\}$. We can similarly show that (\ref{eq:dp}) holds on $A_2=A_1\cup \{\md^*_1\}$ for any $\md^*_1\in A\setminus A_1$. By repeating this process (picking $\md^*_i\in A\setminus A_i$), we can show that (\ref{eq:dp}) holds on $A_{i+1}=A_i\cup\{\md^*_i\}$ for each $i$.
	
	Since $A$ is finite ($D^n$ is finite), this process will eventually terminate when $A_i=A$, \emph{i.e.} $i=|A\setminus A_0|$. Hence, (\ref{eq:dp}) will hold on $A$ as required.
	\end{proof}

We now look at a simple example to demonstrate the impact of Theorem~\ref{th:dexpsuff1}.

	\begin{example}[$\ell^1$ norm]\label{eg:th2}
	For this example, we consider a discrete exponential mechanism where $D=\{0,1,2\}, n=2$ and $u(\md, \md')=-\|\md-\md'\|_1$. In this case, $|D^n|=9$, and we are therefore required to check $2^9-2=510$ subsets (all subsets of $D^n$ except $\emptyset$ and $D^n$ itself) for every pair of neighbouring databases $\md\sim\md'\in D^n$, meaning a total of $510\times36=18,360$ checks.

	Let $\md=\binom{0}{1}$ and $\md'=\binom{2}{1}$, then $\S=\left\{\binom{0}{0},\binom{0}{1},\binom{0}{2}\right\}$. Hence, for this particular pair of neighbouring databases, it is sufficient to check that (\ref{eq:dp}) holds on just $2^3-1=7$ subsets (all subsets of $\S$ except $\emptyset$).

	If we choose $\md=\binom{1}{1}$ and $\md'=\binom{2}{1}$, then we get $\S=\left\{\binom{0}{0}, \binom{0}{1}, \binom{0}{2}, \binom{1}{0}, \binom{1}{1}, \binom{1}{2}\right\}$. This gives a total of $2^6-1=63$ subsets to check.

	By populating the entire set of databases, we can show that 21 pairs of neighbour databases require 7 subset checks, while the remaining 15 pairs require 63 checks. That leaves us with a total of 1,092 subset checks to verify differential privacy, compared with 18,360 without the use of Theorem~\ref{th:dexpsuff1}.
	\end{example}

\subsection{Response mechanism with fixed $C_\md$}

By Theorem~\ref{th:dexpsuff1}, we know that, for a discrete exponential mechanism, checking that (\ref{eq:dp}) holds on all subsets of $\S$ is equivalent to checking all subsets of $D^n$. However, if $C_\md$ is fixed for all $\md\in D^n$, we can partition $\S$ to reduce our workload further. Let us first consider the following set relating to $\S$:
\begin{equation}
\alpha =\{u(\md,\md^*)-u(\md',\md^*)\mid\md^*\in\S\}.
\end{equation}
There are only finitely many such values in $\alpha$, since $\S$ is finite. We label these elements $\alpha_1, \alpha_2, \dots, \alpha_s$, with $0<\alpha_1<\alpha_2<\cdots<\alpha_s$.

We then partition $\S$ into the collection of subsets $\{\St^1,\St^2,\dots,\St^s\}$ as follows:
\begin{equation}
\St^i=\left\{\md^*\in\S\mid u\left(\md,\md^*\right)-u\left(\md',\md^*\right)=\alpha_i\right\}.
\end{equation}

Note that for each $i$ and for all $\md^*\in\St^i$,
\begin{align}\label{eq:delf}
\Prob(X_{\md'}=\md^*)&=Ce^{u(\md', \md^*)}\nonumber\\
&= Ce^{u(\md, \md^*)-\alpha_i}\nonumber\\
&=e^{-\alpha_i}\Prob(X_\md=\md^*).
\end{align}

We can now show that if (\ref{eq:dp}) holds on these partitions, it will hold on all subsets of each partition, i.e. the partitions are sufficient sets of themselves.

\begin{theorem}\label{th:gensuffsets}
Fix $\md\sim\md'\in D^n$ and let $\{X_\md\}$ be a discrete exponential mechanism with $C_\md=C$ for all $\md\in D^n$. If (\ref{eq:dp}) holds on $\St^i$ then it will hold on all $A\subseteq\St^i$ for each $i$.
\end{theorem}

	\begin{proof}
	Fix $\md\sim\md'\in D^n$. We assume
	$$\Prob(X_\md\in\St^i)\le e^\epsilon\Prob(X_{\md'}\in\St^i)+\delta,$$
	and let $A\subseteq\St^i$. By (\ref{eq:delf}), $\Prob(X_{\md'}\in A)=e^{-\alpha_i}\Prob(X_\md\in A)$. Since this also holds for the set $\St^i$ itself, we have
	$$1\le e^{\epsilon-\alpha_i}+\frac{\delta}{\Prob(X_\md\in\St^i)}.$$

	Clearly $\Prob(X_\md\in A)\le \Prob(X_\md\in\St^i)$, which gives
	$$1\le e^{\epsilon-\alpha_i}+\frac{\delta}{\Prob(X_\md\in A)},$$
	or, rewriting,
	$$\Prob(X_\md\in A)\le e^\epsilon\Prob(X_{\md'}\in A)+\delta.$$

	Hence, (\ref{eq:dp}) holds on all $A\subseteq\St^i$.
	\end{proof}

When $\delta=0$, Theorem~\ref{th:gensuffsets} tells us that these partitions $\St^i$ of $\S$ are all we need to check to verify $\epsilon$-differential privacy, i.e. the sufficient sets are the collection of subsets $\{\St^1, \dots, \St^s\}$.

\begin{corollary}[Sufficient sets with fixed $C_\md$]\label{cr:nonrelaxeddp}
Let $\{X_\md\}$ be a discrete exponential mechanism, $\delta=0$ and fix $\md\sim\md'\in D^n$. If (\ref{eq:dp}) holds on $\St^i$ for all $i$, then it will hold on all subsets of $D^n$.
\end{corollary}

	\begin{proof}
	Let $A\subseteq\S$ and assume that (\ref{eq:dp}) holds on $\St^i$ for all $i$. Hence, by Theorem~\ref{th:gensuffsets}, (\ref{eq:dp}) holds on all subsets of $\St^i$ for all $i$. As $\{\St^i\}$ partitions $\S$, $A=\bigcup_{i}A\cap\St^i$, hence $\Prob(X_\md\in A\cap\St^i)\le e^\epsilon\Prob(X_{\md'}\in A\cap\St^i)$ for all $i$.
	
	Since $\St^i\cap\St^j=\emptyset$ when $i\ne j$, $\Prob(X_\md\in\bigcup_{i}A\cap\St^i)=\sum_i \Prob(X_\md\in A\cap\St^i)$, and so,
	\begin{align*}
	\Prob(X_\md\in A)&=\Prob\left(X_\md\in\bigcup_{i}A\cap\St^i\right)\\
	&=\sum_i \Prob\left(X_\md\in A\cap\St^i\right)\\
	&\le e^\epsilon\sum_i \Prob\left(X_{\md'}\in A\cap\St^i\right)\\
	&= e^\epsilon\Prob(X_{\md'}\in A).
	\end{align*}
	Therefore (\ref{eq:dp}) holds on all subsets of $\S$, and by Theorem~\ref{th:dexpsuff1}, it holds on all subsets of $D^n$.
	\end{proof}

\subsection{Discrete exponential mechanism with hamming distance}

The usefulness of Theorem~\ref{th:gensuffsets} becomes particularly apparent when we restrict our response mechanism to one derived from hamming distance. For the remainder of this section our utility function $u$ is defined to be
\begin{equation}\label{eq:demham}
u(\md, \md') = -k\h(\md,\md'),
\end{equation}
where $k\ge0$ is a privacy parameter and $\h(\md,\md')$ is the hamming distance between $\md$ and $\md'$ (the number of elements on which they differ, see (\ref{eq:ham})). Note that for this set-up, the normalisation constant $C_\md=C$ is fixed for all $\md\in D^n$.

\begin{proposition}
For a discrete exponential mechanism satisfying (\ref{eq:demham}),
\begin{equation}
C_\md=C=\left(1+\frac{m}{e^k}\right)^{-n}.
\end{equation}
\end{proposition}

	\begin{proof}
	Let $\md\in D^n$ be given. Then,
	\begin{align*}
	\sum_{\md^\prime\in D^n}\Prob(X_\md=\md^\prime) &=\sum_{\md^\prime\in D^n} C_\md e^{-k\h(\md, \md^\prime)}\\
	&=C_\md e^0+C_\md n m e^{-k}+C_\md \binom{n}{2}m^2e^{-2k}\\
	&\qquad\qquad+\cdots+C_\md m^ne^{-nk}\\
	&=C_\md\sum_{i=0}^n\binom{n}{i}m^ie^{-ik}\\
	&=C_\md\left(1+\frac{m}{e^k}\right)^n.
	\end{align*}

	For the probability mass function of $X_\md$ to sum to 1, we need
	$$C_\md\left(1+\frac{m}{e^k}\right)^n=1.$$
	Rearranging this completes the proof.
	\end{proof}

In this case, the sufficient sets for the problem condense down to a single set.

\begin{corollary}[Sufficient sets for hamming distance]\label{cr:dexpsuff}
Consider a discrete exponential mechanism satisfying (\ref{eq:demham}). If (\ref{eq:dp}) holds on $\S$, then it will hold on all $A\subseteq D^n$.
\end{corollary}

	\begin{proof}
	First note that, for all $\md\sim\md'\in D^n$, $$\h(\md',\md^*)-\h(\md,\md^*)\in\{-1,0,1\},$$ hence for all $\md^*\in\S$, $\h(\md',\md^*)-\h(\md,\md^*)=1$. As $U$ is linear, the set $\alpha$ reduces to a singleton set,
	\begin{align*}
	\alpha&=\{-k\h(\md,\md^*)+k\h(\md',\md^*)\mid\md^*\in\S\}\\
	&=\{k\},
	\end{align*}
	and hence $\S=\St^1$.

	We can then conclude that if (\ref{eq:dp}) holds on $\S$, it must hold on all subsets of $\S$ (by Theorem~\ref{th:dexpsuff1}) and also on all subsets of $D^n$ (by Theorem~\ref{th:gensuffsets}).
	\end{proof}

We have established that, for the discrete exponential mechanism satisfying (\ref{eq:demham}) and for a given neighbouring pair of databases $\md\sim\md'$, to check for differential privacy on all possible subsets of the database space $D^n$, we need only check a single set $\S$.

It is now a relatively simple task to establish conditions on the response mechanism for differential privacy. For the following theorem, we assume that $\delta<1$ (note that all mechanisms are trivially ($\epsilon$,1)-differentially private).

\begin{theorem}[Condition for differential privacy]\label{th:dexp}
Let $\delta<1$. A discrete exponential mechanism $\{X_\md\}$ satisfying (\ref{eq:demham}) is {\ab}-differentially private if and only if
\begin{equation}\label{eq:dp2}
e^k\le\frac{e^\epsilon+m\delta}{1-\delta}.
\end{equation}
\end{theorem}

	\begin{proof}
	Fix $\md\sim\md'\in D^n$. By Corollary~\ref{cr:dexpsuff} we need only satisfy (\ref{eq:dp}) on $\S$ for it to hold on all $A\subseteq D^n$. Hence we need
	$$\Prob(X_\md\in\S)\le e^\epsilon \Prob(X_{\md'}\in\S)+\delta.$$
	By definition, $\Prob(X_{\md'}\in\S)=e^{-k}\Prob(X_\md\in\S)$, therefore the mechanism will be differentially private if and only if
	\begin{equation}\label{eq:th5}
	1\le e^{\epsilon-k}+\frac{\delta}{\Prob(X_\md\in\S)}.
	\end{equation}

	Now consider $\Prob(X_\md\in\S)$. By definition, each $\md^*\in\S$ lies a hamming distance of $\h(\md, \md^*)+1$ from $\md'$. Therefore, the number of $\md^*$ where $\h(\md, \md^*)=c$ is $\binom{n-1}{c}m^c$ ($c$ elements must be changed, but not the element on which $\md$ and $\md'$ differ, to ensure $\h(\md, \md^*)+1=\h(\md', \md^*)$). Hence,
	\begin{align*}
	\Prob(X_\md\in\S) &= C\sum_{i=0}^{n-1}\binom{n-1}{i}m^i e^{-ik}\\
	&=C\left(1+\frac{m}{e^k}\right)^{n-1}\\
	&=\left(1+\frac{m}{e^k}\right)^{-1}.
	\end{align*}
	Substituting this result into (\ref{eq:th5}) gives
	$$1\le e^{\epsilon-k}+\delta\left(1+\frac{m}{e^k}\right),$$
	and solving for $e^k$ completes the proof.
	\end{proof}

\textbf{Remark:} For ($\epsilon$,0)-differential privacy, we require $k\le\epsilon$.

\textbf{Discussion:} If we convert our discrete exponential mechanism back into the form of the exponential mechanism, McSherry and Talwar \cite{MT07} tell us that the mechanism satisfies $2\epsilon$-differential privacy at worst. However, we have shown in Theorem~\ref{th:dexp} that the mechanism satisfies $\epsilon$-differential privacy, and that this condition is tight (necessary and sufficient, hence we can do no better). This improves on the looser bound in McSherry and Talwar's proof, as it underestimates the differential privacy achieved by a factor of two. Their mechanism is also limited to $\epsilon$-differential privacy only ($\delta=0$), whereas we can account for {\ab}-differential privacy ($\delta>0$).

We define the error of a mechanism, $\E$, to be the largest mean hamming distance between the input and sanitised databases:
\begin{equation}\label{eq:err}
\E = \max_{\md\in D^n} \mathbb{E}[\h(X_\md, \md)].
\end{equation}
Using the differential privacy constraints established in Theorem~\ref{th:dexp}, we can now determine error bounds of the mechanism.

\begin{theorem}[Error]\label{th:demerr}
The error, $\E$, of a discrete exponential mechanism satisfying (\ref{eq:demham}) and which is {\ab}-differentially private, satisfies
\begin{equation}
\frac{1-\delta}{1+\frac{e^\epsilon}{m}}\le \frac{\E}{n}\le\frac{m}{m+1}.
\end{equation}
\end{theorem}

	\begin{proof}
	Let $\md\in D^n$. Then $\Prob(X_\md=\md')=C e^{-k\h(\md, \md')}$ and $|\{\md': \h(\md, \md')=c\}|=\binom{n}{c}m^c$ for all $c\in\{0, 1, \dots, n\}$. Hence,
	\begin{align*}
	\mathbb{E}[\h(X_\md, \md)] &= C \sum_{i=0}^n i \binom{n}{i} m^i e^{-ik}\\
	&= Cn \sum_{i=1}^n \binom{n-1}{i-1}\left(\frac{m}{e^k}\right)^i\\
	&= Cn\frac{m}{e^k} \sum_{I=0}^{n-1}\binom{n-1}{I} \left(\frac{m}{e^k}\right)^I\\
	&= Cn \frac{m}{e^k} \left(1+\frac{m}{e^k}\right)^{n-1}\\
	&= n \frac{m}{e^k} \left(1+\frac{m}{e^k}\right)^{-1}\\
	&= \frac{n}{1+\frac{e^k}{m}}.
	\end{align*}
	
	On average, the number of entries that will change is therefore
	$$\frac{\E}{n} =\frac{\max_{\md\in D^n} \mathbb{E}[\h(X_\md, \md)]}{n}= \frac{1}{1+\frac{e^k}{m}}.$$
	
	Since the response mechanism is differentially private, Theorem~\ref{th:dexp} tells us that $e^k\le \frac{e^\epsilon +m\delta}{1-\delta}$, and $k\ge 0$ by definition, hence,
	$$\frac{1-\delta}{1+\frac{e^\epsilon}{m}}\le \frac{\E}{n}\le\frac{m}{m+1}.$$
	\end{proof}

\textbf{Remark:} This lower bound bound is tight and can be achieved by setting $k=\ln\left(\frac{e^\epsilon +m\delta}{1-\delta}\right)$ (see Theorem~\ref{th:dexp}).

\section{Product Sanitisation}\label{sc:san}

The results of Section~\ref{sc:dem} give a clear framework on how to create differentially private mechanisms for any type of discrete data, particularly categorical data using hamming distance, and to obtain tight (necessary and sufficient) conditions for differential privacy. However, this mechanism requires the creation of a separate probability distribution for each of the $(m+1)^n$ unique databases.

In this section, we present a simple method for realising the discrete exponential mechanism.

\subsection{Response mechanism}\label{sc:prod}
A product sanitisation mechanism is one where the database is sanitised row-by-row. The following is the definition of such a mechanism, following the same notation as in \cite{HLM15}.

\begin{definition}[Product sanitisation mechanism]\label{df:prod}
Given a 1-dimensional response mechanism for the identity query, $\{X_d:\Omega\to D\mid d\in D\}$ (the parent mechanism), the \emph{product sanitisation response mechanism} is defined to be the set of measurable mappings
$$\{X_\md:\Omega\to D^n \mid \md\in D^n\}$$
given by
\begin{equation}\label{eq:ind}
X_\md=(X_\md^1, \dots, X_\md^n),
\end{equation}
where the $X_\md^i$ are independent and each $X_\md^i$ has the same distribution as $X_{d_i}$, for all $\md\in D^n, i\in \{1,\dots,n\}$.
\end{definition}

In this setting, each row of the sanitised database is represented by an independent random variable, as if the database represents $n$ 1-dimensional databases, each sanitised independently. Realising this framework therefore only requires the creation of $m+1$ probability distributions which are then copied to create distributions on each database.

We now recall Theorem~5 of \cite{HLM15}, which states that differential privacy on a product sanitisation mechanism is guaranteed when its parent mechanism is differentially private.

\begin{theorem}\label{th:prod}
Consider a family $\{X_d\mid d\in D\}$ of measurable mappings and assume that
$$\Prob(X_d\in A)\le e^\epsilon\Prob(X_{d^\prime}\in A)+\delta,$$
for all $d, d^\prime\in D$ and all $A\subseteq D$. Let $\{X_\md\mid\md\in D^n\}$ be a product sanitisation mechanism. Then
$$\Prob(X_\md\in A)\le e^\epsilon\Prob(X_{\md'}\in A)+\delta,$$
for all $\md\sim\md'\in D^n$ and all $A\subseteq D^n$.
\end{theorem}

The converse of this result was proven in Lemma~3 of \cite{HLM15}. However, we are able to take a simpler approach to showing this converse by exploiting the finiteness of $D$.

\begin{corollary}[Parent mechanism]\label{cr:prod}
A product sanitisation mechanism $\{X_\md\mid\md\in D^n\}$ is {\ab}-differentially private if and only if its parent mechanism $\{X_d\mid d\in D\}$ is {\ab}-differentially private.
\end{corollary}

	\begin{proof}
	``$\Leftarrow$'': Theorem~\ref{th:prod}.

	``$\Rightarrow$'': Let $A\subseteq D$ and $d\ne d^\prime\in D$ be given. Define $A^\prime=A\times D^{n-1}\subseteq D^n$ and $\md,\md'\in D^n$, such that $d_1=d$, $d^\prime_1=d^\prime$ and $d_i=d^\prime_i$ for all $i\in\{2, \dots, n\}$, hence $\md\sim\md'$. Since $\{X_\md\}$ is differentially private by assumption,
	$$\Prob(X_\md\in A^\prime)\le e^\epsilon\Prob(X_{\md'}\in A^\prime)+\delta.$$
	However, since $\{X_\md\}$ is a product sanitisation mechanism,
	\begin{align*}
	\Prob(X_\md\in A^\prime)  &= \Prob(X_\md^1\in A)\times\prod_{i=2}^n\Prob(X_\md^i\in D)\\
	&=\Prob(X_\md^1\in A)\\
	&=\Prob(X_{d_1}\in A).
	\end{align*}
	Hence,
	$$\Prob(X_d\in A) \le e^\epsilon\Prob(X_{d^\prime}\in A)+\delta,$$
	for all $d\ne d^\prime\in D$ and $A\subseteq D$, as required.
	\end{proof}

Hence, for a product sanitisation mechanism to be differentially private, we need only show that its parent mechanism is differentially private.

For the remainder of this section, given $0\le p\le\frac{1}{m+1}$, the parent mechanism $\{X_d\mid d\in D\}$ of the product sanitisation is defined such that
\begin{equation}\label{eq:uni}
\Prob(X_d=d) = 1-pm, \quad \Prob(X_d=d^\prime) = p,
\end{equation}
for every $d^\prime\in D\setminus \{d\}$.

\textbf{Remark:} $p=\frac{1}{m+1}$ represents the case of releasing uniform noise (i.e. no information), while decreasing $p$ reduces the error of the mechanism.

\textbf{Remark:} By requiring $p\le\frac{1}{m+1}$, we get $\Prob(X_d=d)\ge\Prob(X_d=d^\prime)$ for every $d^\prime\in D$. Therefore, the mechanism is at least as likely to return the correct answer as any one incorrect answer, an entirely reasonable assumption.

	\begin{example}[Categorical data II]\label{eg:prod}
	Using the same set-up as in Example~\ref{eg:prelim}, one such permissable $p$ would be $p=0.1$, since $0\le0.1\le\frac{1}{5}$. Then, if $d$ were to be the value `Television', and $d^\prime$ the value `Cars', $\Prob(X_d=d)=0.6$, and $\Prob(X_d=d^\prime)=0.1$.
	
	We then sanitise the $n$-row database $\md$ in the same way, working through the database one row at a time.
	\end{example}

In the first main result of this section, we show that the product sanitisation mechanism and the discrete exponential mechanism are equivalent, despite being constructed in different ways (this is subject to the discrete exponential mechanism satisfying (\ref{eq:demham}) and the product sanitisation mechanism satisfying (\ref{eq:uni})).

\begin{theorem}[Equivalence]\label{th:equiv}
Let $\{X_\md\mid\md\in D^n\}$ be a discrete exponential mechanism satisfying (\ref{eq:demham}), and $\{Y_\md\mid\md\in D^n\}$ be a product sanitisation response mechanism, whose parent mechanism satisfies (\ref{eq:uni}). Then the probability mass functions of $X_\md$ and $Y_\md$ are identical, for every $\md\in D^n$, when
\begin{equation}\label{eq:equiv}
e^k=\frac{1}{p}-m.
\end{equation}
\end{theorem}

	\begin{proof}
	Let $\md\in D^n$, then
	\begin{align*}
	\Prob(X_\md=\md') &= \left(1+\frac{m}{e^k}\right)^{-n}e^{-k\h(\md, \md')},\\
	\Prob(Y_\md=\md') &= (1-pm)^n\left(\frac{p}{1-pm}\right)^{\h(\md, \md')},
	\end{align*}
	for all $\md'\in D^n$.
	
	For the two mechanisms to be equivalent, we need $\frac{p}{1-pm}=e^{-k}$, or, rewriting, $\frac{1-pm}{p}=e^k$.

	We also need $1-pm=\frac{e^k}{e^k+m}$, which can be rewritten as $\frac{1}{1-pm}=1+\frac{m}{e^k}$, or $\frac{1-pm}{p}=e^k$.

	Hence,  the mechanisms$\{X_\md\}$ and $\{Y_\md\}$ are equivalent when $e^k=\frac{1}{p}-m$ or $p=\frac{1}{e^k+m}$.
	\end{proof}

\subsection{Alternative proofs of Theorems~\ref{th:dexp} and~\ref{th:demerr}}
We have already established that the discrete exponential mechanism satisfying (\ref{eq:demham}) and the product sanitisation mechanism satisfying (\ref{eq:uni}) are equivalent, meaning the results of Theorems~\ref{th:dexp} and \ref{th:demerr} also apply to the product sanitisation mechanism in this particular set-up. In this sub-section, we provide alternative methods of proof for these theorems that make use of the specific product sanitisation of the mechanism, for the additional insight provided.

The proof of the following theorem first appeared in \cite{HLM15}, but is included here for completeness.

\begin{theorem}[Condition for differential privacy]\label{th:uni}
A product sanitisation mechanism, whose parent mechanism satisfies (\ref{eq:uni}), is {\ab}-differentially private if and only if
\begin{equation}\label{eq:unidp}
p\ge \frac{1-\delta}{e^\epsilon+m}.
\end{equation}
\end{theorem}

	\begin{proof}
	By Corollary~\ref{cr:prod}, we need only be concerned with proving differential privacy on $\{X_d\mid d\in D\}$ for it to hold on $\{X_\md\mid\md\in D^n\}$.

	``$\Rightarrow$'': Assume $\{X_d\}$ is differentially private and let $d\ne d^\prime\in D$ be given. Applying the definition of differential privacy to the singleton set $A=\{d\}$ gives$$1-pm\le e^\epsilon p+\delta.$$Rearranging this gives (\ref{eq:unidp}).
	
	``$\Leftarrow$'': Assume (\ref{eq:unidp}) holds and let $d\ne d^\prime\in D$ and $A\subseteq D$ be given. There are four cases to consider on $A$:
	\begin{enumerate}
	\item $d, d^\prime \notin A$: Then $\Prob(X_d\in A) = \Prob(X_{d^\prime}\in A) = p|A|$ and differential privacy holds trivially.
	\item $d, d^\prime \in A$: Then $\Prob(X_d\in A) = \Prob(X_{d^\prime}\in A) = p(|A|-1)+1-pm=p(|A|-m-1)+1$ and differential privacy holds trivially.
	\item $d^\prime \in A, d\notin A$: Then $\Prob(X_d\in A) = p|A|$ and $\Prob(X_{d^\prime}\in A) = p(|A|-1)+1-pm$, so $\Prob(X_d\in A)\le\Prob(X_{d^\prime}\in A)$ and differential privacy holds.
	\item $d \in A, d^\prime\notin A$: Then \begin{align*}\Prob(X_d\in A) &= p(|A|-m-1)+1,\\ \Prob(X_{d^\prime}\in A)&=p|A|.\end{align*}
	From (\ref{eq:unidp}), we have
	\begin{align*}
	1-pm&\le e^\epsilon p+\delta\\
	&= e^\epsilon (p|A|-p|A|+p)+\delta\\
	&\le e^\epsilon p|A| - p(|A|-1)+\delta,
	\end{align*}
	since $|A|\ge1$ ($d\in A$ by hypothesis). Rewriting the above,
	$$p(|A|-m-1)+1 \le e^\epsilon \left(p|A|\right)+\delta,$$
	hence $\Prob(X_d\in A)\le e^\epsilon \Prob(X_{d^\prime}\in A)+\delta$ and differential privacy holds.
	\end{enumerate}
	\end{proof}

\textbf{Remark:} For ($\epsilon$,0)-differential privacy, we require $p\ge\frac{1}{e^\epsilon+m}$.

\textbf{Remark:} By definition, $p$ is bounded from above by $\frac{1}{m+1}$, so for differential privacy we require
$$\frac{1-\delta}{e^\epsilon+m}\le p\le\frac{1}{m+1}.$$

We now look at an alternative proof of Theorem~\ref{th:demerr} which established error bounds on the mechanism.

\begin{theorem}[Error]\label{th8}
The error, $\E$, of a product sanitisation mechanism, whose parent mechanism satisfies (\ref{eq:uni}) and is {\ab}-differentially private, satisfies
\begin{equation}
\frac{1-\delta}{1+\frac{e^\epsilon}{m}}\le \frac{\E}{n}\le\frac{m}{m+1}.
\end{equation}
\end{theorem}

	\begin{proof}
	Let $\md\in D^n$. Then,
	\begin{align*}
	\mathcal{E}&=\max_{\md\in D^n}\mathbb{E}\left[\h(X_\md, \md)\right]\\
	&= \max_{\md\in D^n}\mathbb{E}\left[\sum_{i=1}^n \h(X_\md^i, d_i)\right]\\
	&=\max_{\md\in D^n}\sum_{i=1}^n \mathbb{E}\left[\h(X_\md^i, d_i)\right]\\
	&=\max_{\md\in D^n}\sum_{i=1}^n \Prob(X_{d_i}\ne d_i)\\
	&=n\left(\max_{d\in D}\left(1-\Prob(X_d=d)\right)\right)\\
	&= npm.
	\end{align*}
	
	On average, the number of entries that will change is therefore
	$$\frac{\E}{n}=pm.$$
	
	As the mechanism satisfies {\ab}-differential privacy, $p\ge \frac{1-\delta}{e^\epsilon+m}$ by Theorem~\ref{th:uni}, and since $p\le\frac{1}{m+1}$ by definition,
	$$\frac{1-\delta}{1+\frac{e^\epsilon}{m}}\le \frac{\E}{n}\le\frac{m}{m+1}.$$
	\end{proof}

\textbf{Remark:} As with the discrete exponential mechanism, this bound is tight and can be achieved by setting $p$ equal to the lower bound established in Theorem~\ref{th:uni}.

\subsection{Optimal mechanism}

We now look at the second main result of this section. Using the same error definition as before, (\ref{eq:err}), we show how to construct the optimal {\ab}-differentially private mechanism which produces the minimum error. For the purpose of this subsection, we assume the following labelling of elements in the data set:
\begin{equation}
D = \{\dt_1, \dt_2, \dots, \dt_{m+1}\}.
\end{equation}

\begin{definition}[Solution Matrix]
The parent mechanism $\{X_d\}$ of a product sanitisation mechanism can be defined by a stochastic matrix $$P_{(\epsilon,\delta)}\in[0,1]^{(m+1)\times(m+1)}, $$ where
\begin{equation}
\Prob(X_{\dt_i}=\dt_j)=\left[P_{(\epsilon,\delta)}\right]_{ij}.
\end{equation}
We refer to $P_{(\epsilon,\delta)}$ as a product sanitisation \emph{solution matrix} if the mechanism it defines is {\ab}-differentially private.
\end{definition}

\textbf{Remark:} A parent mechanism which satisfies (\ref{eq:uni}) is represented by the following solution matrix:
\begin{equation}\label{eq:solnmtx}
\left[P_{(\epsilon,\delta)}\right]_{ij} = \begin{cases}
1-mp &\text{if } i = j,   \\
p &\text{if } i \ne j.   \end{cases}
\end{equation}

\begin{theorem}[Optimality]\label{th:opt}
Let $\{X_d\}$ be a parent mechanism which satisfies (\ref{eq:uni}), with $p=\frac{1-\delta}{e^\epsilon+m}$. The error, $\E=\max_{\md\in D^n}\mathbb{E}[h(X_\md,\md)]$, produced by its product sanitisation mechanism $\{X_\md\}$ is the minimum of all product sanitisation mechanisms.
\end{theorem}

	\begin{proof}
	Let $A=P_{(\epsilon,\delta)}$, satisfying (\ref{eq:solnmtx}) with $p=\frac{1-\delta}{e^\epsilon+m}$, be the solution matrix of $\{X_d\}$, \emph{i.e.} $\Prob(X_{\dt_i}=\dt_j)=a_{ij}$. Note that this mechanism has the property that $a_{jj}=e^\epsilon a_{ij}+\delta$, since $a_{jj}=\frac{e^\epsilon+m\delta}{e^\epsilon+m}$ and $a_{ij}=\frac{1-\delta}{e^\epsilon+m}$, for all $i\ne j$.
	
	The error of its product sanitisation mechanism $\{X_\md\}$, using the same method as in the proof of Theorem~\ref{th8}, is
	\begin{align*}
	\max_{\md\in D^n}\mathbb{E}[\h(X_\md, \md)] &= n\left(\max_{d\in D}\left(1-\Prob(X_d=d)\right)\right)\\
	&=n\left(\max_i(1-a_{ii})\right)\\
	&=n\left(1-\min_i a_{ii}\right)\\
	&=n(1-a_{ii}),
	\end{align*}
	for all $i$, since $a_{ii}=a_{jj}=1-mp$ for all $i$ and $j$.
	
	Let $B$ be a solution matrix defining the parent mechanism $\{Y_d\}$, with a corresponding product sanitisation mechanism $\{Y_\md\}$, where $B\ne A$. Since $B$ is a solution matrix, it is stochastic ($\sum_j b_{ij}=1$ for all $i$) and the parent mechanism it defines is {\ab}-differentially private ($\Prob(Y_d\in A)\le e^\epsilon \Prob(Y_{d'}\in A)+\delta$, for all $d,d'\in D$ and $A\subseteq D$).
	
	Since $A\ne B$ and $A,B$ are stochastic, there exists at least one pair $(i^*, j^*)$, where $b_{i^*j^*}< a_{i^*j^*}$. There are two cases to consider:
	\begin{enumerate}
	\item $i^*=j^*$: The error of $\{Y_\md\}$ is then
	\begin{align*}
	\max_{\md\in D^n}\mathbb{E}[\h(Y_\md, \md)]&=n\left(1-\min_i b_{ii}\right)\\
	&\ge n\left(1-b_{j^*j^*}\right)\\
	&>n(1-a_{j^*j^*})\\
	&=\max_{\md\in D^n}\mathbb{E}[\h(X_\md, \md)].
	\end{align*}

	\item $i^*\ne j^*$: As noted previously, $a_{j^*j^*}=e^\epsilon a_{i^*j^*}+\delta$, and since $\{Y_d\}$ is {\ab}-differentially private, $\Prob(Y_{d_{j^*}}=d_{j^*})\le e^\epsilon\Prob(Y_{d_{i^*}}=d_{j^*})+\delta$, or alternatively, $b_{j^*j^*}\le e^\epsilon b_{i^*j^*}+\delta$. Therefore,
	\begin{align*}
	b_{j^*j^*}&\le e^\epsilon b_{i^*j^*}+\delta\\
	&<e^\epsilon a_{i^*j^*}+\delta\\
	&=a_{j^*j^*}.
	\end{align*}

	Hence, $b_{j^*j^*}<a_{j^*j^*}$, and case 1 applies.
	\end{enumerate}
	
	We therefore conclude that $$\max_{\md\in D^n}\mathbb{E}[\h(Y_\md, \md)]>\max_{\md\in D^n}\mathbb{E}[\h(X_\md, \md)],$$ and that $\{X_\md\}$ is the product sanitisation mechanism which produces the optimal error.
	\end{proof}

Using Theorem~\ref{th:opt}, we now have a simple method to construct the most accurate {\ab}-differentially private mechanism possible for discrete data. We have proven that no other product sanitisation mechanism is more accurate than it. It follows from Theorem~\ref{th:equiv} that we now know the optimal discrete exponential mechanism that satisfies (\ref{eq:demham}).

\section{Conclusion}\label{sc:conc}

We study mechanisms for {\ab}-differential privacy on finite datasets that are an adaptation for discrete data of the exponential mechanism introduced by McSherry and Talwar.   By deriving \emph{sufficient sets} for differential privacy we obtain necessary and sufficient conditions for differential privacy, a tight lower bound on the maximal expected error of a discrete mechanism and a characterisation of the optimal mechanism which minimises the maximal expected error within the class of mechanisms considered.



\section*{Acknowledgments}
The first named author was supported by the Science Foundation Ireland grant SFI/11/PI/1177.

\end{document}